\documentclass[11pt]{article}
\usepackage[margin=1.1in]{geometry}
\usepackage{amsmath,amssymb,amsthm,mathtools}
\usepackage{booktabs,array}
\usepackage{graphicx}
\usepackage{tikz}
\usetikzlibrary{positioning}
\usepackage{pgfplots}
\pgfplotsset{compat=1.18}
\usepackage[colorlinks=true,linkcolor=blue,citecolor=blue,urlcolor=blue]{hyperref}
\usepackage{xcolor}
\usepackage{placeins}

\newtheorem{theorem}{Theorem}
\newtheorem{lemma}[theorem]{Lemma}
\newtheorem{proposition}[theorem]{Proposition}
\newtheorem{corollary}[theorem]{Corollary}
\newtheorem{definition}[theorem]{Definition}
\newtheorem{remark}[theorem]{Remark}

\newcommand{\PC}{\mathrm{PC}}
\newcommand{\LCC}{\mathrm{LCC}}
\newcommand{\NC}{\mathrm{NC}}
\newcommand{\dam}{\mathrm{dam}}
\newcommand{\Top}{\mathrm{Top}}

\newcommand{\tautr}{\tau^{+}}

\title{Objective Transfer in Single-Node Network Criticality:\\
Optimal Constants and Exact Separation Thresholds}

\author{Onur U\u{g}urlu\thanks{Department of Computer Engineering, Izmir Bakircay University, Izmir, T\"urkiye. ORCID: 0000-0003-2743-5939. E-mail: \texttt{onur.ugurlu@bakircay.edu.tr}.}}
\date{September 2026}

\begin{document}
\maketitle

\begin{abstract}
When a vertex is selected as the most damaging single attack target under
one damage objective, how much of the optimal damage under a different
objective does it retain? We study this question for the three objectives
commonly used in critical node detection --- pairwise connectivity ($\PC$),
size of the largest surviving component ($\LCC$) and number of components
($\NC$) --- for single-vertex attacks on connected graphs, with the
attacked vertex counted in the damage and ties resolved optimistically.
For the pair $\PC$, $\LCC$ the loss is bounded by constants: a
$\PC$-optimal vertex retains at least a $2-\sqrt2$ fraction of the optimal
$\LCC$ damage, an $\LCC$-optimal vertex at least $2/3$ of the optimal
$\PC$ damage, and both constants are best possible. Neither bound depends
on the choice among tied optimal vertices or on whether the attacked
vertex is counted, and both remain valid as lower bounds for attack sets
of any fixed size. For the four directions involving $\NC$ no positive
constant exists; explicit families have ratios decaying like $1/k$, and no
single attack set retains a positive fraction of all three optima
uniformly. We also determine the smallest orders at which the sets of
optimal vertices become disjoint: $7$ for $\LCC/\NC$, $8$ for $\PC/\NC$,
$9$ for $\PC/\LCC$ and $11$ for all three pairwise, with trees realising
the triple separation for every $n\ge11$; these values rest on an
exhaustive enumeration of the $11{,}989{,}762$ connected graphs with
$3\le n\le10$. Finally, the most critical vertex in a Birnbaum-type sense
may depend on the failure probability: a six-vertex graph, the smallest
possible order, changes leader exactly once, at $p^*=3-\sqrt5$. Adding
universal vertices lifts the budget-one examples to every fixed budget, so
over all connected graphs the same constants are optimal at every fixed
budget; the behaviour within restricted classes such as trees is left
open. 
\end{abstract}

% =====================================================================
\section{Introduction}\label{sec:intro}

Identifying the most critical vertex of a network is a basic step in
protection budgeting, interdiction and repair prioritisation. The word
``critical'' presupposes an objective. In the critical node detection
literature the damage caused by removing a set of vertices is measured,
depending on the application, by the number of disconnected pairs, by
the size of the largest surviving component, or by the number of
components, and the three optimisation problems have mostly been treated
separately \cite{arulselvan2009cndp,lalou2018survey,walteros2019cnp}.
Bi-objective formulations are known to differ from their single-objective
projections \cite{ventresca2018biobjective}, some algorithmic frameworks
handle several objectives side by side \cite{aringhieri2016framework},
and the connected variant has been studied for all three measures within
one model \cite{hosteins2022connected}. What appears to be missing is a
bound on the damage lost when a vertex computed as optimal for one
objective is evaluated under another. The comparisons we are aware of
study the objectives jointly, or compare their optimal solutions and
related rankings empirically on real-world networks
\cite{ugurlu2022centrality,ugurlu2025reliability}, without a directed
worst-case ratio. The question arises whenever a generic
criticality score is reused across applications, and it is the subject
of this paper.

We consider single-vertex attacks on connected graphs and three damage
functions: pairwise connectivity ($\PC$), largest surviving component
($\LCC$) and component count ($\NC$). For two objectives $u,w$ the
\emph{transfer ratio} of a graph is the fraction of the optimal
$w$-damage achieved by a $u$-optimal vertex. Two conventions are fixed
throughout and matter for the numerical values: damages count the
attacked vertex as lost (Definition~\ref{def:objectives}), and ties
within the source objective are resolved optimistically, in favour of
the best vertex among the tied optima (Definition~\ref{def:transfer}).
The optimistic choice makes negative results stronger; for the positive
results it turns out to be immaterial, since the same bounds hold for
every tied optimum.

The results are of three kinds. For the pair $\PC$, $\LCC$ the transfer
ratio is bounded below by a constant: at least $2-\sqrt2$ from a
$\PC$-optimal vertex to the $\LCC$ objective, and at least $2/3$ in the
opposite direction. Both constants are attained in the limit by explicit
families and are therefore best possible (Theorem~\ref{thm:main}). They
hold for every tied source optimum and are unchanged if damages are
measured on the survivors only (Theorem~\ref{thm:tauminus}), and as lower
bounds they carry over to attack sets of any fixed size
(Corollary~\ref{cor:budget}). For the four directions involving $\NC$
no such constant exists: along explicit families the ratio tends to zero
at rates such as $6k/(k^2+3k+1)$ and $1/k$, and no single attack set
retains a positive fraction of all three optima uniformly
(Proposition~\ref{prop:qsplit}). Adding universal vertices to the
budget-one examples shows that over all connected graphs this
classification is the same at every fixed budget
(Section~\ref{sec:apex}); what remains open is the behaviour inside
classes that the construction does not reach, such as trees for $b\ge2$.

The second kind of result concerns small graphs. The sets of optimal
vertices of $\LCC$ and $\NC$ first become disjoint at order $7$, those
of $\PC$ and $\NC$ at order $8$, those of $\PC$ and $\LCC$ at order $9$,
and all three become pairwise disjoint first at order $11$, on a tree
(Figure~\ref{fig:triple}); trees with three distinct optimal vertices
exist for every $n\ge11$ (Theorems~\ref{thm:thresholds}
and~\ref{thm:W}). The third result concerns the failure probability
rather than the objective: in a Birnbaum-type reliability setting the
most critical vertex of a six-vertex graph changes exactly once as the
failure probability varies, at $p^*=3-\sqrt5$, and no smaller graph
shows such a change (Theorem~\ref{thm:regime}).

The finite-size separation thresholds rest on exhaustive enumeration;
Section~\ref{sec:method} describes the exact computational methods and
summarises the computations supporting the results.

\subsection{Related work}\label{sec:related}

The critical node detection problem and its variants optimise pairwise
connectivity, the largest component or the number of components
directly \cite{arulselvan2009cndp,lalou2018survey,walteros2019cnp}, and
\cite{ventresca2018biobjective} shows that the bi-objective problem
differs from both projections. Early attack-vulnerability studies
already observed that different damage measures respond differently to
targeted removals \cite{holme2002attack}. Algorithms have been developed
for each objective separately, including distributed algorithms for
pairwise connectivity \cite{ugurlu2023adapc} and for the number of
components \cite{ugurlu2025maxnum}, and for the largest component and
the number of partitions restricted to articulation points
\cite{akram2023articulation}. The optimal solutions of the three
problems have been compared empirically on real-world networks
\cite{ugurlu2025reliability}, and centrality measures have been compared
as proxies for the component-count objective \cite{ugurlu2022centrality}.
We propose no optimisation algorithm; the object of study is the
worst-case behaviour of one objective's optimiser when evaluated under
another, which complements such empirical comparisons.

That component-importance rankings can change with the common
reliability parameter is known in reliability theory
\cite{zhu2012birnbaum,cai2013birnbaum}, severity-dependent rank
crossings have been observed in epidemic models \cite{qu2017severity},
and reliability polynomials of graphs may cross
\cite{colbourn1993reliability,kelmans2000crossing}. Relative to this
literature, Theorem~\ref{thm:regime} adds exact minimality and a
closed-form breakpoint. Axiomatic treatments characterise centrality
measures themselves \cite{boldi2014axioms}; here we compare the
\emph{decisions} induced by damage-based objectives, with conventions
that are insensitive to ties and automorphisms. Simultaneous
approximation --- one solution near-optimal for several measures at
once --- has been studied in other models
\cite{stein1997bicriteria,azar2002allnorm,goel2006simultaneous,goel2005majorization,schoot2025simultaneous};
Section~\ref{sec:qsplit} concerns the canonical network damages, where
the component count is not a norm. For the computational parts we follow
the practice of certifying algorithms \cite{mcconnell2011certifying}:
each reported result is accompanied by a witness that can be checked
independently.

% =====================================================================
\section{Definitions and conventions}\label{sec:defs}

\begin{definition}\label{def:objectives}
Let $G=(V,E)$ be a finite simple connected graph, $n=|V|\ge 3$, and let
$S\subseteq V$ with $|S|=b$, $1\le b\le n-1$, be an \emph{exact-cardinality
attack}. Write $\sigma(G-S)=(c_1\ge\dots\ge c_m)$ for the component profile.
The damages are
\[
\dam_{\PC}(S)=\tbinom n2-\textstyle\sum_i\tbinom{c_i}2,\qquad
\dam_{\LCC}(S)=n-c_1,\qquad
\dam_{\NC}(S)=\max\{m-1,0\},
\]
extended to arbitrary $S\subseteq V$ (as needed for the random attack sets
of Definition~\ref{def:regime}) with the convention $c_1=0$ and $m=0$ for
the empty remainder.
\end{definition}

\begin{definition}\label{def:transfer}
$\Top_u(G,b)=\arg\max_{|S|=b}\dam_u(S)$ (the set of \emph{all} maximisers;
for $b=1$ we identify the singleton $\{v\}$ with the vertex $v$).
Objectives $u,w$ \emph{disagree} on $(G,b)$ iff
$\Top_u\cap\Top_w=\emptyset$; since champion sets are unions of
automorphism orbits, this convention is both tie-safe and
automorphism-safe. For $M_w(G,b)=\max_{|S|=b}\dam_w(S)>0$, the
\emph{optimistic transfer ratio} is
\[
\tautr_{u\to w}(G,b)=\max_{S\in\Top_u(G,b)}\dam_w(S)/M_w(G,b)\in[0,1].
\]
A direction $u\to w$ is \emph{uniformly positive} if
$\inf_{G}\tautr_{u\to w}(G,b)\ge\rho$ for a constant $\rho>0$ independent of
$G$ and $n$, and \emph{vanishing} if the infimum is $0$. The
\emph{pessimistic} ratio $\tau^-_{u\to w}$ is defined with $\min$ in
place of $\max$ over the source champion set.
\end{definition}

\begin{definition}\label{def:regime}
For $p\in(0,1)$ let $F_{-v}\subseteq V\setminus\{v\}$ contain each vertex
independently with probability $p$. The \emph{Birnbaum-type signed marginal
effect} of $v$ under objective $u$ is
\[
C_{u,p}(v)=\mathbb E\bigl[\dam_u(F_{-v}\cup\{v\})\bigr]
          -\mathbb E\bigl[\dam_u(F_{-v})\bigr],
\]
computed exactly (for $u=\NC$ this quantity may be negative). The
\emph{regime champion set} is
$\Top^{\mathrm{reg}}_{u,p}(G)=\arg\max_{v\in V}C_{u,p}(v)$. All regime
statements below are proved via the exact polynomials $C_{u,p}(v)$
with isolated real roots; numerical grids are used for plotting only.
\end{definition}

\begin{remark}\label{rem:normalisation}
The damages of Definition~\ref{def:objectives} count the attacked vertices
as lost. The natural survivor-based alternatives are, for exact budget
$b$ and $N=n-b$ survivors, the constant shifts
$\dam'_\PC=\dam_\PC-\bigl[\binom n2-\binom N2\bigr]$ and
$\dam'_\LCC=\dam_\LCC-b$ (for $b=1$: $n-1$ and $1$).
Such shifts leave every champion set --- and hence all separation
thresholds of Theorem~\ref{thm:thresholds} --- unchanged, but they do
change the transfer \emph{ratios} and may create zero denominators. The
constants of Theorem~\ref{thm:main} are stated for the conventions of
Definition~\ref{def:objectives}; Theorem~\ref{thm:tauminus} shows that
under the survivor-based shifts the same two constants are optimal at
$b=1$, and Corollary~\ref{cor:budget} that both lower bounds persist at
every exact budget. For other functionals and normalisations nothing is
claimed.
\end{remark}

% =====================================================================
\section{Transfer ratios at budget one}\label{sec:classification}

\begin{theorem}\label{thm:main}
For $b=1$ the worst-case optimistic transfer ratios are exactly
\[
\begin{array}{c|ccc}
u\backslash w & \NC & \PC & \LCC\\\hline
\NC & - & 0 & 0\\
\PC & 0 & - & 2-\sqrt2\\
\LCC & 0 & 2/3 & -
\end{array}
\]
i.e.\ every direction involving $\NC$ is vanishing, with explicit witness
families of rate $\Theta(1/k)$, while
$\inf_G\tautr_{\PC\to\LCC}=2-\sqrt2$ and $\inf_G\tautr_{\LCC\to\PC}=2/3$,
both attained in the limit by explicit families.
\end{theorem}

The four vanishing directions are witnessed by the families of
Propositions~\ref{prop:F1}--\ref{prop:F4}: $H(k,2)$ with rate
$6k/(k^2+3k+1)$, $H(k,k/2)$ with rate $2/k$, $D(k)$ with rate $1/k$, and
the hub-and-tail family with rate $2(k+1)/(k^2+k+2)$. The two lower
bounds are Propositions~\ref{prop:pc-lcc} and~\ref{prop:lcc-pc}, and
their optimality is Proposition~\ref{prop:sharp}; all proofs are in
Section~\ref{sec:proofs}.

The two positive constants do not depend on how ties in the source
objective are broken, nor on whether the attacked vertex is counted in
the damage.

\begin{theorem}\label{thm:tauminus}
At exact budget one, restricted to cells where the target admits
positive damage, and under either damage convention, the optimal
worst-case pessimistic transfer constants coincide with the optimistic
ones: $2-\sqrt2$ for $\PC\to\LCC$, $2/3$ for $\LCC\to\PC$, and $0$ for
all four directions involving $\NC$. The two positive guarantees hold
uniformly for every source maximiser.
\end{theorem}

The proof (Section~\ref{sec:proof-tauminus}) reuses the envelope argument
of Section~\ref{sec:proofs-constants}, which never selects a particular
champion. Section~\ref{sec:budgets} extends the lower bounds to attack
sets of any fixed size and shows that, over all connected graphs, the
whole table of Theorem~\ref{thm:main} is the same at every fixed budget.

% =====================================================================
\section{Separation thresholds}\label{sec:thresholds}

Disagreement between the objectives has exact small-graph thresholds.

\begin{theorem}\label{thm:thresholds}
For $b=1$:
(i) $\Top_{\LCC}$ and $\Top_{\NC}$ intersect on every connected graph with
$n\le 6$ (for $n=3$ this is immediate: in $K_3$ all champion sets are
equal, in $P_3$ the centre is the common champion; the cases
$4\le n\le6$ are checked exhaustively), and are disjoint on exactly $8$
of the $853$ connected graphs with $n=7$. One of them is the centred
double star $D(2)$ of Section~\ref{sec:proof-D}: its centre is the
unique $\LCC$-champion (damage $4$) while the two hubs are the
$\NC$-champions (damage $2$);
(ii) $\Top_{\PC}$ and $\Top_{\NC}$ intersect on every connected graph with
$n\le 7$, and are disjoint on exactly $48$ of the $11{,}117$ connected
graphs with $n=8$; one of them consists of two hubs with two pendant
leaves each, joined by a path through two further vertices, and there
the two middle vertices are the $\PC$-champions while the two hubs are
the $\NC$-champions;
(iii) $\Top_{\PC}$ and $\Top_{\LCC}$ intersect on every connected graph
with $n\le 8$, and are disjoint on exactly $129$ of the $261{,}080$
connected graphs with $n=9$; one of them is the centred double star
$D(3)$, whose centre ($\LCC$-champion, profile
$(4,4)$) differs from its two hubs ($\PC$-champions, profile
$(5,1,1,1)$);
(iv) \emph{(triple separation)} no connected graph with $n\le 10$ separates
all three champion sets pairwise simultaneously (by exhaustive
enumeration of the $11{,}716{,}571$ connected graphs with $n=10$), whereas \emph{for every} $n\ge 11$ there is a
\emph{tree} on $n$ vertices whose three champions are unique and pairwise
distinct: the explicit backbone families $W(\alpha)$ ($n=2\alpha+5$, odd
$n\ge11$) and $W^{\mathrm{ev}}(\alpha)$ ($n=2\alpha+6$, even $n\ge16$),
together with sporadic trees for $n\in\{12,14\}$ ($3$ of $551$,
resp.\ $34$ of $3159$ trees). The triple-separation threshold is exactly
$11$.
\end{theorem}

Parts (i)--(iii) and the finite part of (iv) are results of the
enumeration described in Section~\ref{sec:method}; the infinite part of
(iv) is Theorem~\ref{thm:W}, proved in Section~\ref{sec:proof-W}.
Figure~\ref{fig:triple} shows the smallest example.

\begin{figure}[t]
\centering
\begin{tikzpicture}[
  every node/.style={font=\small},
  vert/.style={circle,draw,inner sep=1.6pt,minimum size=5.5mm},
  champN/.style={vert,fill=orange!35,thick},
  champL/.style={vert,fill=cyan!30,thick},
  champP/.style={vert,fill=magenta!25,thick},
  leaf/.style={circle,draw,inner sep=1.2pt,minimum size=4mm,fill=gray!12}]
  % backbone: h(6) - b(8) - c(4) - p(10) - r(0)
  \node[champN] (h) at (0,0) {$h$};
  \node[vert]   (b) at (1.7,0) {$b$};
  \node[champL] (c) at (3.4,0) {$c$};
  \node[champP] (p) at (5.1,0) {$p$};
  \node[vert]   (r) at (6.8,0) {$r$};
  \draw (h)--(b)--(c)--(p)--(r);
  % leaves
  \node[leaf] (h1) at (-1.0,0.75) {}; \node[leaf] (h2) at (-1.25,0) {};
  \node[leaf] (h3) at (-1.0,-0.75) {};
  \draw (h)--(h1) (h)--(h2) (h)--(h3);
  \node[leaf] (p1) at (5.1,-1.05) {}; \draw (p)--(p1);
  \node[leaf] (r1) at (7.65,0.7) {}; \node[leaf] (r2) at (7.65,-0.7) {};
  \draw (r)--(r1) (r)--(r2);
  % annotations
  \node[anchor=north,text width=3.1cm,align=center] at (0,-1.35)
    {$\NC$ champion\\profile $(7,1,1,1)$,\ $\dam_\NC=3$};
  \node[above=5mm of c,text width=3.0cm,align=center]
    {$\LCC$ champion\\profile $(5,5)$,\ $\dam_\LCC=6$};
  \node[anchor=north,text width=3.2cm,align=center] at (5.1,-1.35)
    {$\PC$ champion\\profile $(6,3,1)$,\ $\dam_\PC=37$};
\end{tikzpicture}
\caption{The smallest graph on which all three champions differ pairwise:
an $11$-vertex tree ($W(3)$ up to isomorphism). Deleting $h$ shatters four
components; deleting $c$ bisects the tree; deleting $p$ trades balance for
an extra component and wins pairwise connectivity. No graph on at most
$10$ vertices admits such a triple separation
(Theorem~\ref{thm:thresholds}(iv)), and trees realising it exist for every
$n\ge11$ (Theorem~\ref{thm:W}).}
\label{fig:triple}
\end{figure}
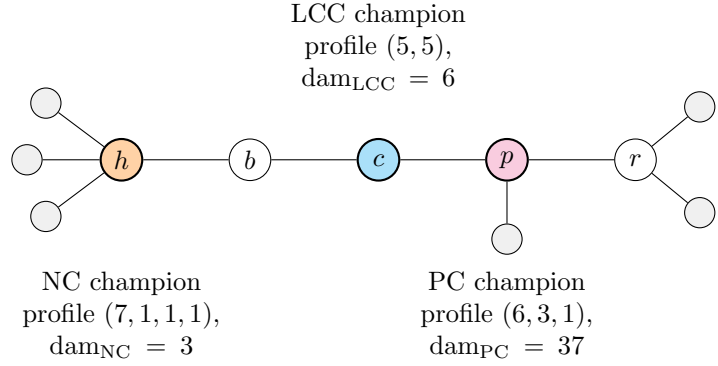

% =====================================================================
\section{Dependence on the failure probability}\label{sec:regime}

\begin{theorem}\label{thm:regime}
Among all connected graphs with $n\le 5$ no leadership change of the
Birnbaum-type $\PC$-criticality occurs anywhere in $p\in(0,1)$ (by exact
root isolation). Let $R_6$ be the graph on $\{v_0,\dots,v_5\}$
formed by the triangle $v_0v_1v_2$ and the $4$-cycle $v_0v_2v_3v_4$,
which share the edge $v_0v_2$, together with a pendant vertex $v_5$
attached to $v_1$. For $R_6$ the criticality polynomials satisfy
$C_1(p)-C_0(p)=(1-p)^2(p^2-6p+4)$, producing a unique leadership change at
\[
p^{*}=3-\sqrt5\approx 0.7639,
\]
with $\Top^{\mathrm{reg}}_{\PC,p}=\{v_1\}$ for $p<p^{*}$,
$\{v_0,v_1,v_2\}$ at $p=p^{*}$, and $\{v_0,v_2\}$ for $p>p^{*}$.
\end{theorem}

\begin{proof}
For each of the $29$ connected graphs with $n\le5$, all
non-identically-zero pairwise differences $C_{u,p}(v)-C_{u,p}(w)$ are
polynomials with rational coefficients; isolating their real roots in
$(0,1)$, evaluating the champion set at exact rational points inside each
resulting subinterval, \emph{and} evaluating it exactly (in the relevant
algebraic number field) at each isolated root itself, shows that no
leadership handover occurs anywhere in $(0,1)$. For $R_6$ the six
exact polynomials are as plotted in Figure~\ref{fig:regime}; by symmetry
$C_0\equiv C_2$ and $C_3\equiv C_4$, the difference
$C_1-C_0=(1-p)^2(p^2-6p+4)$ has the single interior root $p^*=3-\sqrt5$,
and among all non-identically-zero pairwise differences this is the only
root in $(0,1)$ --- so the leadership envelope changes exactly once.
Evaluating at rational points on either side of $p^*$ (e.g.\ $p=2/5$ and
$p=22/25$) and exactly at $p^*\in\mathbb Q(\sqrt5)$ identifies the three
champion sets stated. The difference polynomials, isolated roots and
evaluation points are given in the supplementary material.
\end{proof}

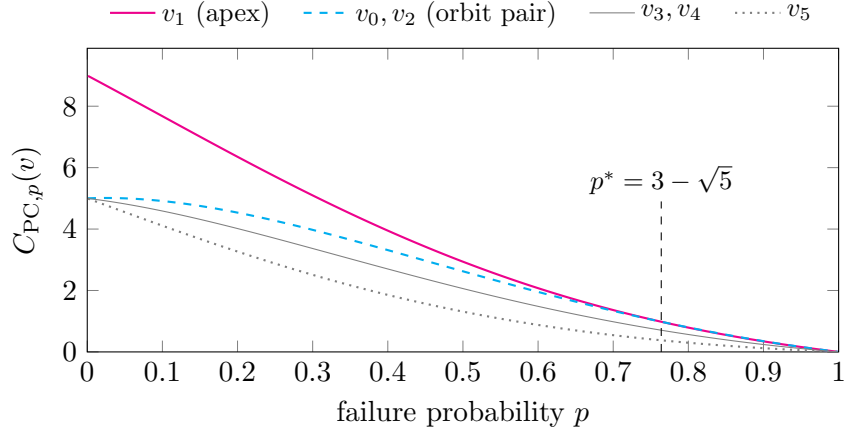
\begin{figure}[t]
\centering
\begin{tikzpicture}
\begin{axis}[width=0.72\textwidth,height=5.6cm,
  xlabel={failure probability $p$},
  ylabel={$C_{\PC,p}(v)$},
  xmin=0,xmax=1,ymin=0,
  legend style={font=\small,at={(0.5,1.04)},anchor=south,draw=none,
    /tikz/every even column/.append style={column sep=10pt}},
  legend columns=4,
  legend cell align=left,
  domain=0.001:0.999,samples=140]
\addplot[thick,magenta]
  {2*x^5-10*x^4+16*x^3-4*x^2-13*x+9};
\addlegendentry{$v_1$ (apex)}
\addplot[thick,cyan,dashed]
  {2*x^5-11*x^4+24*x^3-21*x^2+x+5};
\addlegendentry{$v_0,v_2$ (orbit pair)}
\addplot[gray]
  {2*x^5-10*x^4+19*x^3-13*x^2-3*x+5};
\addlegendentry{$v_3,v_4$}
\addplot[gray,dotted,thick]
  {2*x^5-8*x^4+10*x^3-9*x+5};
\addlegendentry{$v_5$}
\addplot[black,dashed,thin,forget plot] coordinates {(0.76393,0) (0.76393,4.9)};
\node[font=\small] at (axis cs:0.76393,5.6) {$p^*=3-\sqrt5$};
\end{axis}
\end{tikzpicture}
\caption{Exact Birnbaum-type $\PC$-criticality polynomials of the
graph $R_6$ of Theorem~\ref{thm:regime}. The
leadership changes exactly once, at $p^*=3-\sqrt5$: the apex $v_1$
leads for small failure probabilities, the orbit pair $\{v_0,v_2\}$ for
large ones.}
\label{fig:regime}
\end{figure}

% =====================================================================
% Proofs for budget one (witness families, the two constants, ties and
% survivor scale) -- separate file.
% proofs.tex -- budget one: witness families, the two constants, tied
% optima and the survivor scale.

% ---------------------------------------------------------------------
\section{Proofs for budget one}\label{sec:proofs}\label{sec:proofs-families}

Throughout, deleting a vertex $v$ from the connected graph $G$ leaves the
component profile $\pi(v)=(c_1\ge\dots\ge c_m)$ with $\sum_i c_i=n-1$, and we
write $Q(v)=\sum_i\binom{c_i}{2}$, so that
\[
\dam_\PC(v)=\tbinom n2-Q(v),\qquad
\dam_\LCC(v)=n-c_1,\qquad
\dam_\NC(v)=m-1 .
\]
(For single-vertex deletions from a connected graph on $n\ge3$ vertices
the remainder is non-empty, so $m\ge1$ and the truncation in
Definition~\ref{def:objectives} is inactive.)
Maximising $\dam_\PC$ is equivalent to minimising $Q$. In every table below
the profile entries sum to $n-1$; leaf deletions always leave a single
component and are never competitive, so we list them once per family.

% ---------------------------------------------------------------------
\subsection{The two-clique family $H(k,\ell)$: the cells
$\NC\to\PC$ and $\PC\to\NC$}\label{sec:proof-H}

\begin{definition}
For $k\ge2$, $\ell\ge1$, the graph $H(k,\ell)$ consists of two disjoint
cliques $K_k$; a distinguished vertex $c$ of the first clique carries
$\ell$ pendant leaves; a second distinguished vertex $x\neq c$ of the first
clique is joined by a bridge to a vertex $y$ of the second clique. Thus
$n=2k+\ell$.
\end{definition}

\begin{lemma}\label{lem:H-profiles}
The deletion profiles are:
\[
\begin{array}{l|l|l}
v & \pi(v) & Q(v)\\\hline
c & (2k-1,\,1^{\ell}) & \binom{2k-1}{2}\\
x & (k+\ell-1,\,k) & \binom{k+\ell-1}{2}+\binom k2\\
y & (k+\ell,\,k-1) & Q(x)+\ell\\
\text{other clique vertices} & (n-1) & \binom{n-1}{2}\\
\text{leaves} & (n-1) & \binom{n-1}{2}
\end{array}
\]
and consequently
\(
\dam_\NC(c)=\ell,\ \dam_\NC(x)=\dam_\NC(y)=1,
\)
all other vertices having $\dam_\NC=0$. Moreover
\begin{equation}\label{eq:H-key}
2\bigl[Q(c)-Q(x)\bigr]=2k^2-2k\ell-\ell^2-2k+3\ell .
\end{equation}
\end{lemma}

\begin{proof}
Deleting $c$ isolates the $\ell$ leaves; the remainder (both cliques joined
by the bridge, minus $c$) stays connected of order $2k-1$. Deleting $x$
splits off the second clique (order $k$) from the first clique with the
leaves (order $k+\ell-1$); deleting $y$ is analogous with orders $k+\ell$
and $k-1$, and $Q(y)-Q(x)=\ell$ by a one-line computation. All remaining
deletions leave a connected graph. Identity~\eqref{eq:H-key} is direct
expansion.
\end{proof}

\begin{proposition}\label{prop:F1}
Let $\ell=2$ and $k\ge3$. Then $\Top_\NC=\{c\}$ and $\Top_\PC=\{x\}$, both
unique, the sets are disjoint, and
\[
\tautr_{\NC\to\PC}\bigl(H(k,2)\bigr)=\frac{6k}{k^2+3k+1}
=\Theta(1/k)\longrightarrow0 .
\]
At $k=2$ the construction breaks: \eqref{eq:H-key} becomes negative and
$\Top_\PC=\{c\}$, destroying disjointness; hence $k\ge3$ is sharp.
\end{proposition}

\begin{proof}
With $\ell=2$, \eqref{eq:H-key} equals $2(k^2-3k+1)>0$ for $k\ge3$, so $x$
uniquely minimises $Q$ among $\{c,x,y\}$; the remaining vertices have the
strictly larger value $\binom{n-1}2$, since with $\ell=2$, $n=2k+2$,
$\binom{n-1}2-Q(c)=\binom{2k+1}2-\binom{2k-1}2=4k-1>0$. Uniqueness of $c$ in $\NC$ is
immediate from Lemma~\ref{lem:H-profiles}. Finally
$\dam_\PC(c)=\binom n2-\binom{2k-1}2=6k$ and
$M_\PC=\dam_\PC(x)=k^2+3k+1$ by direct computation with $n=2k+2$.
\end{proof}

\begin{proposition}\label{prop:F2}
Let $k\ge4$ be even and $\ell=k/2$. Then $\Top_\PC=\{x\}$ uniquely,
$\Top_\NC=\{c\}$ uniquely, the sets are disjoint, and
\[
\tautr_{\PC\to\NC}\bigl(H(k,k/2)\bigr)=\frac1\ell=\frac2k
\longrightarrow0 .
\]
\end{proposition}

\begin{proof}
With $\ell=k/2$, \eqref{eq:H-key} equals $k(3k-2)/4>0$, so $x$ is again the
unique $Q$-minimiser (the comparison with the non-cut vertices is as
before). Here $\dam_\NC(c)=\ell\ge2>1=\dam_\NC(x)$ gives unique champions
and $\tautr=\dam_\NC(x)/\dam_\NC(c)=1/\ell$.
\end{proof}

% ---------------------------------------------------------------------
\subsection{The centred double star $D(k)$: the cell $\LCC\to\NC$ and
the constant $2/3$}\label{sec:proof-D}

\begin{definition}
For $k\ge2$, $D(k)$ has a centre $c$ adjacent to two hubs $h_1,h_2$, each
hub carrying $k$ pendant leaves; $n=2k+3$.
\end{definition}

\begin{proposition}\label{prop:F3}
In $D(k)$:
\[
\begin{array}{l|l|l|l}
v & \pi(v) & \dam_\PC(v) & (\dam_\LCC,\dam_\NC)\\\hline
c & (k+1,k+1) & (k+1)(k+3) & (k+2,\,1)\\
h_i & (k+2,\,1^k) & (k+1)(3k+4)/2 & (k+1,\,k)\\
\text{leaf} & (2k+2) & 2k+2 & (1,\,0)
\end{array}
\]
Hence for all $k\ge2$: $\Top_\LCC=\{c\}$ and $\Top_\NC=\{h_1,h_2\}$, and
\[
\tautr_{\LCC\to\NC}\bigl(D(k)\bigr)=\frac1k\longrightarrow0,
\qquad
\tautr_{\LCC\to\PC}\bigl(D(k)\bigr)=\frac{2(k+3)}{3k+4}\longrightarrow\frac23 .
\]
Moreover $\dam_\PC(h_i)-\dam_\PC(c)=(k+1)(k-2)/2$, so
$\Top_\PC=\{h_1,h_2\}$ for $k\ge3$ while $\Top_\PC=\{c,h_1,h_2\}$ at $k=2$
(where the second ratio correctly evaluates to $1$). At $k=3$ ($n=9$) the
two ratios are $1/3$ and $12/13$, which are also the smallest values of
these two directed ratios over all connected graphs of order $9$
(Table~\ref{tab:comp}).
\end{proposition}

\begin{proof}
Deleting $c$ leaves two stars of order $k+1$; deleting a hub isolates its
$k$ leaves and leaves the star of the other hub, joined through $c$, of
order $k+2$; leaf deletion is neutral. All claims are direct computations
from the table; disjointness $\{c\}\cap\{h_1,h_2\}=\emptyset$ requires
$k\ge2$ for $\NC$ (at $k=1$ the centre ties).
\end{proof}

% ---------------------------------------------------------------------
\subsection{The hub-and-tail family: the cell $\NC\to\LCC$}\label{sec:proof-tail}

\begin{definition}
For $k\ge2$ let $T(k)$ consist of a path on $L=k^2$ vertices
$p_1,\dots,p_L$, a hub $h$ adjacent to $p_L$, and $k$ pendant leaves on
$h$; $n=L+k+1$.
\end{definition}

\begin{proposition}\label{prop:F4}
In $T(k)$ (where $L=k^2$ counts path \emph{vertices}), $\Top_\NC=\{h\}$
uniquely, and
\[
\tautr_{\NC\to\LCC}\bigl(T(k)\bigr)
=\frac{k+1}{(L+k)/2+1}
=\frac{2(k+1)}{k^2+k+2}
=\Theta(1/k)\longrightarrow0 .
\]
\end{proposition}

\begin{proof}
$\dam_\NC(h)=k$ (its leaves shatter), every internal path vertex has
$\dam_\NC=1$, and endpoints and leaves have $\dam_\NC\le1$; since $k\ge2$,
$h$ is the unique $\NC$-champion. Deleting the path vertex $p_j$
($1\le j\le L$) leaves exactly two relevant components, of orders $j-1$
and $L-j+k+1$; hence $c_1(p_j)=\max\{j-1,\,L-j+k+1\}$ is minimised at the
balance point $j_*=(L+k+2)/2$, which is a uniquely attained integer
because $L+k=k^2+k$ is even, and lies in the admissible range
$1\le j_*\le L$ since $L=k^2\ge k+2$ for $k\ge2$. There $c_1=(L+k)/2$ and
$M_\LCC=n-c_1=(L+k)/2+1$ (deleting $h$ gives only $c_1=L$, i.e.\
$\dam_\LCC(h)=k+1$, and no other vertex does better than $p_{j_*}$).
Finally $\tautr=\dam_\LCC(h)/M_\LCC$ as displayed.
\end{proof}

% ---------------------------------------------------------------------
\subsection{Trees with three distinct champions: the families $W(\alpha)$ and
$W^{\mathrm{ev}}(\alpha)$}\label{sec:proof-W}

\begin{definition}
For $\alpha\ge3$, $W(\alpha)$ consists of the backbone path
$h\,b\,c\,p\,r$ with $\alpha$ pendant leaves on $h$, $\alpha-2$ on $p$ and
$2$ on $r$; $n=2\alpha+5$. For $\alpha\ge5$, $W^{\mathrm{ev}}(\alpha)$ adds
one pendant leaf on $c$; $n=2\alpha+6$.
\end{definition}

\begin{theorem}\label{thm:W}
For every $\alpha\ge3$ the three champions of $W(\alpha)$ are unique and
pairwise distinct:
\[
\Top_\PC=\{p\},\qquad \Top_\LCC=\{c\},\qquad \Top_\NC=\{h\},
\]
with $\dam_\NC(h)=\alpha$, $\dam_\LCC(c)=\alpha+3$ and
$\dam_\PC(p)=\binom n2-\binom{\alpha+3}2-3$. The same champion identities
hold for $W^{\mathrm{ev}}(\alpha)$ with $\alpha\ge5$, where
$\dam_\NC(h)=\alpha$, $\dam_\LCC(c)=\alpha+4$ and
$\dam_\PC(p)=\binom n2-\binom{\alpha+4}2-3$. Together with the
sporadic trees for
$n\in\{12,14\}$ this realises triple separation on a tree for every
$n\ge11$; both parameter thresholds are sharp.
\end{theorem}

\begin{proof}
\emph{$W(\alpha)$.} The deletion profiles are
\[
\begin{array}{l|l|l|l}
v & \pi(v) & (\dam_\LCC,\dam_\NC) & Q(v)\\\hline
h & (\alpha+4,\,1^{\alpha}) & (\alpha+1,\,\alpha) & \binom{\alpha+4}2\\
b & (\alpha+3,\,\alpha+1) & (\alpha+2,\,1) & \alpha^2+3\alpha+3\\
c & (\alpha+2,\,\alpha+2) & (\alpha+3,\,1) & \alpha^2+3\alpha+2\\
p & (\alpha+3,\,3,\,1^{\alpha-2}) & (\alpha+2,\,\alpha-1) & \binom{\alpha+3}2+3\\
r & (2\alpha+2,\,1,\,1) & (3,\,2) & 2\alpha^2+3\alpha+1\\
\text{leaves} & (2\alpha+4) & (1,\,0) & \binom{2\alpha+4}2
\end{array}
\]
The $\LCC$ and $\NC$ columns immediately give unique champions $c$
($\alpha+3$ beats $\alpha+2,\alpha+2,\alpha+1,3,1$) and $h$
($\alpha$ beats $\alpha-1,2,1,1,0$), using $\alpha\ge3$. For $\PC$ the
factorised differences
\[
\begin{aligned}
Q(h)-Q(p)&=\alpha,&
Q(b)-Q(p)&=\tfrac{(\alpha+3)(\alpha-2)}2,\\
Q(c)-Q(p)&=\tfrac{\alpha^2+\alpha-8}2,&
Q(r)-Q(p)&=\tfrac{(3\alpha-5)(\alpha+2)}2.
\end{aligned}
\]
are all strictly positive for $\alpha\ge3$ (and leaves are dominated by
inspection), so $p$ is the unique $Q$-minimiser.

\emph{Sharpness at $\alpha=2$.} Removing the leaves at $p$ yields
backbone $Q$-values $(15,13,12,13,15)$ and
$\dam_\NC$-values $(2,1,1,1,2)$: both the $\PC$-championship of $p$ and the
uniqueness of $h$ in $\NC$ fail simultaneously.

\emph{$W^{\mathrm{ev}}(\alpha)$.} With the extra leaf on $c$ the
deletion profiles become $h:(\alpha+5,1^{\alpha})$,
$b:(\alpha+4,\alpha+1)$, $c:(\alpha+2,\alpha+2,1)$,
$p:(\alpha+4,3,1^{\alpha-2})$, $r:(2\alpha+3,1,1)$, leaves:
$(2\alpha+5)$; in particular $\dam_\LCC(c)=\alpha+4$ and
$\dam_\NC(c)=2$, and the $\LCC$/$\NC$ columns again single out $c$ and
$h$. For $\PC$ the factorised differences are
\[
Q(h)-Q(p)=\alpha+1,\quad
Q(b)-Q(p)=\tfrac{(\alpha+3)(\alpha-2)}2,\quad
Q(c)-Q(p)=\tfrac{\alpha^2-\alpha-14}2,
\]
\[
Q(r)-Q(p)=\tfrac{3(\alpha^2+\alpha-4)}2,\quad
Q(\ell)-Q(p)=\tfrac{3\alpha^2+11\alpha+2}2 ,
\]
all strictly positive for $\alpha\ge5$; the threshold-critical one is
$Q(c)-Q(p)$, positive precisely for $\alpha\ge5$. At $\alpha=4$ one
computes $Q(c)=30<31=Q(p)$ and the $\PC$-championship migrates to $c$,
while the $\LCC$ and $\NC$ champions are unaffected --- so
$\alpha\ge5$ is sharp.

\emph{Sporadic even orders.} Exhaustive generation over all $551$ trees on
$12$ vertices and all $3159$ trees on $14$ vertices yields exactly $3$,
resp.\ $34$, triple-separating trees; they are listed in the
supplementary material.
\end{proof}

% ---------------------------------------------------------------------
\subsection{The two constants}\label{sec:proofs-constants}

Throughout this subsection $b=1$, $N=n-1$, and for a vertex $v$ we abbreviate
$s(v)=c_1(v)$, the largest part of its deletion profile.

\begin{lemma}\label{lem:envelope}
For every vertex $v$ with $s=s(v)$,
\[
Q(v)\;\le\;F(N,s):=q\binom s2+\binom r2,
\qquad q=\Bigl\lfloor\frac Ns\Bigr\rfloor,\; r=N-qs,
\]
and in particular
\begin{equation}\label{eq:env}
Q(v)\le
\begin{cases}
N(s-1)/2, & s\le N/2,\\[1mm]
\binom s2+\binom{N-s}2, & s\ge N/2 .
\end{cases}
\end{equation}
Moreover $Q(v)\ge\binom s2$ always.
\end{lemma}

\begin{proof}
Fixing the largest part $s$ and the total $N$, the convex functional
$\sum_i\binom{c_i}2$ is maximised by making the parts as large as allowed,
i.e.\ by $q$ parts of size $s$ and one remainder part; this is $F(N,s)$.
For $s\le N/2$ we relax $F(N,s)\le (N/s)\binom s2=N(s-1)/2$; for
$s\ge N/2$ at most one further part fits and \eqref{eq:env} is exact. The
lower bound keeps only the largest part.
\end{proof}

\begin{proposition}\label{prop:pc-lcc}
Every connected $G$ with a well-defined ratio satisfies
\[
\tautr_{\PC\to\LCC}(G)\ge2-\sqrt2.
\]
\end{proposition}

\begin{proof}
Let $\ell$ be an $\LCC$-champion, $a=s(\ell)=\min_v s(v)$, and let $p$ be a
$\PC$-champion with $B=s(p)\ge a$. Since $p$ minimises $Q$,
\[
\tbinom B2\le Q(p)\le Q(\ell)\le\text{RHS of \eqref{eq:env} at }s=a .
\]
Among the $\PC$-champions choose $p$ with $B=s(p)$ minimal; the optimistic
ratio is attained at a champion with smallest largest part, so
$\tautr_{\PC\to\LCC}=(n-B)/(n-a)$ for this choice.
Set $c=2-\sqrt2$ and $R=n-c(n-a)$. We claim $B\le R$, i.e.\
$B(B-1)\le R(R-1)$. In the regime $a\le N/2$ it suffices that
$R(R-1)\ge N(a-1)$; the difference, viewed as a function of $a$, is
decreasing and at $a=N/2$ equals
$\tfrac32(2-\sqrt2)N+4-3\sqrt2>0$. In the regime $a\ge N/2$, writing
$d=N-a$, the difference $R(R-1)-a(a-1)-d(d-1)$ is concave in $d$ and
positive at both endpoints $d=0$ and $d=N/2$. Hence $B\le R$ and
\[
\tautr_{\PC\to\LCC}
=\frac{n-B}{n-a}\;\ge\;\frac{n-R}{n-a}=2-\sqrt2 .
\qedhere
\]
\end{proof}

\begin{proposition}\label{prop:lcc-pc}
For every connected $G$, $\tautr_{\LCC\to\PC}(G)\ge2/3$.
\end{proposition}

\begin{proof}
With $\ell,a$ as above and $K=\binom n2$, note
$M_\PC=K-\min_vQ(v)\le K-\binom a2$ (every $v$ has $s(v)\ge a$), and
$K-\binom a2=\tfrac{(N+a)(N+1-a)}2$. In the regime $a\le N/2$,
Lemma~\ref{lem:envelope} gives
\[
\frac{\dam_\PC(\ell)}{M_\PC}
\ge\frac{K-N(a-1)/2}{K-\binom a2}
=\frac{N(N+2-a)}{(N+a)(N+1-a)} ,
\]
and $3N(N+2-a)-2(N+a)(N+1-a)=N^2+4N-3Na-2a+2a^2$ is convex in $a$ with
vertex at $a=(3N+2)/4>N/2$, hence decreasing on $[0,N/2]$ with minimum
value $3N>0$ at $a=N/2$. In the regime $a\ge N/2$, with $d=N-a$,
\[
\frac{\dam_\PC(\ell)}{M_\PC}
\ge\frac{K-\binom a2-\binom d2}{K-\binom a2}\ge\frac23
\iff
(2N-d)(d+1)\ge3d(d-1)
\iff
(d+1)N-2d^2+d\ge0,
\]
which holds since $d\le N/2$ implies $2d^2-d\le Nd<N(d+1)$.
\end{proof}

\begin{proposition}\label{prop:sharp}
Both constants are optimal:
\begin{enumerate}
\item ($2/3$) In $D(k)$, $\tautr_{\LCC\to\PC}=2(k+3)/(3k+4)\to2/3$
      (Proposition~\ref{prop:F3}).
\item ($2-\sqrt2$) For $a\ge6$ let
      $B_a=\max\{B:\binom B2<a(a-1)\}$, $r=B_a-a$, $t=2a-B_a$. These
      parameters are admissible: $\binom{a+2}2\le a(a-1)$ holds iff
      $a^2-5a-2\ge0$, i.e.\ for $a\ge6$, so $B_a\ge a+2$ and $r\ge2$;
      and $\binom{2a}2=a(2a-1)>a(a-1)$ gives $B_a<2a$, so $t\ge1$. Build:
      a clique $K_r$ containing $p$ and $u$, with $t$ pendant leaves on
      $p$; a disjoint clique $K_a$ containing $v$; and a connector $z$
      adjacent exactly to $u$ and $v$ (so $n=2a+1$). Then $z$ is the unique
      $\LCC$-champion with profile $(a,a)$, $p$ is the unique
      $\PC$-champion with profile $(B_a,1^t)$, and
      \[
      \tautr_{\PC\to\LCC}=\frac{2a+1-B_a}{a+1}
      \xrightarrow[a\to\infty]{}2-\sqrt2,
      \]
      since $B_a=\sqrt2\,a+O(1)$.
\end{enumerate}
\end{proposition}

\begin{proof}[Proof of (2)]
Deleting $z$ leaves the two cliques with their attachments: components of
orders $a$ (the $K_a$) and $a$ (the $K_r$ with $p$'s $t$ leaves,
$r+t=a$), so $\pi(z)=(a,a)$, $Q(z)=2\binom a2=a(a-1)$ and
$\dam_\LCC(z)=n-a=a+1$. Deleting $u$ or $v$ gives the profile
$(a+1,a-1)$, and every other deletion leaves a component of order at least
$B_a>a+1$ (for $a\ge6$); in all cases $c_1\ge a+1>a$, so $z$ is the unique
$\LCC$-champion. Deleting $p$ isolates its $t$ leaves and leaves the
component $(K_r\setminus p)\cup\{z\}\cup K_a$ of order $(r-1)+1+a=B_a$, so
$\pi(p)=(B_a,1^t)$ and $Q(p)=\binom{B_a}2$. By the choice of $B_a$,
$Q(p)=\binom{B_a}2<a(a-1)=Q(z)$; deleting $u$ or $v$ gives
$Q=\binom{a+1}2+\binom{a-1}2=a^2-a+1>Q(p)$, and the remaining vertex types
leave connected graphs of order $n-1$ with strictly larger $Q$ still.
Hence $p$ is the unique $\PC$-champion,
$\tautr_{\PC\to\LCC}=\dam_\LCC(p)/\dam_\LCC(z)=(n-B_a)/(a+1)$ with
$n=2a+1$, and $\binom{B_a}2<a(a-1)\le\binom{B_a+1}2$ forces
$B_a=\sqrt2\,a+O(1)$, giving the stated limit (e.g.\ the ratio equals
$2/3$ at $a=8$ and $0.595041\ldots$ at $a=120$; see also
Table~\ref{tab:comp}).
\end{proof}

% ---------------------------------------------------------------------
\subsection{Tied optima and the survivor scale}\label{sec:proof-tauminus}

We now prove Theorem~\ref{thm:tauminus}. The survivor-based damages at
budget one are $\dam'_\PC=\dam_\PC-(n-1)$ and $\dam'_\LCC=\dam_\LCC-1$
(Remark~\ref{rem:normalisation}).

\begin{proof}[Proof of Theorem~\ref{thm:tauminus}]
Tie-uniformity is structural. Every $\PC$-champion attains the same
minimum $Q=\sum_i\binom{c_i}2$, and every $\LCC$-champion leaves the
same largest part $a$; the envelope chain
$\binom B2\le Q(p)\le Q(\ell)\le F(N,a)$ (with $F$ the envelope
function of Lemma~\ref{lem:envelope} and $N=n-1$) therefore
applies to each tied champion separately, never invoking a choice among
them.

For the survivor scale in the direction $\PC\to\LCC$, set
$\alpha=2-\sqrt2$ and $R'=N-\alpha(N-a)$. In the regime $a\le N/2$ the
excess $E_1=R'(R'-1)-N(a-1)$ is decreasing in $a$ --- its derivative is
at most $(2\sqrt2-3)N-(2-\sqrt2)<0$ on the regime --- with boundary
value $E_1(N/2)=(2-\sqrt2)N/2>0$. In the regime $a\ge N/2$, writing
$d=N-a$, direct computation gives the identity
\[
R'(R'-1)-a(a-1)-d(d-1)\;=\;
d\bigl[2(\sqrt2-1)(N-2d)+(2-\sqrt2)\bigr]\;\ge\;0 .
\]
In both regimes $B(B-1)\le R'(R'-1)$, hence $B\le R'$ and
$(N-B)/(N-a)\ge2-\sqrt2$ for every tied $\PC$-champion. The
inclusive-scale excesses dominate the survivor ones termwise
($E^{\mathrm{incl}}=E^{\mathrm{surv}}+\beta(2R'+\beta-1)$ with
$\beta=\sqrt2-1$ and $R'\ge1$), so the original convention carries the
same bound.

For $\LCC\to\PC$ on the survivor scale, every tied $\LCC$-champion
$\ell$ satisfies $Q(\ell)\le F(N,a)$ and every vertex leaves largest
part at least $a$, so with $K'=\binom N2$,
\[
\frac{\dam'_\PC(\ell)}{\max_v \dam'_\PC(v)}\;\ge\;
\frac{K'-F(N,a)}{K'-\binom a2}.
\]
In the regime $a\le N/2$ this is at least $N/(N+a-1)\ge2/3$
(equivalent to $N-2a+2\ge0$); in the regime $a\ge N/2$ it suffices that
$K'-\binom a2-3\binom d2=d(N-2d+1)\ge0$. The inclusive branches are the
inequalities of Propositions~\ref{prop:pc-lcc} and~\ref{prop:lcc-pc}
and hold for every tied champion by the same argument.

Optimality: on the sharp families of
Section~\ref{sec:proofs-constants} the source champions are unique, so
pessimistic and optimistic ratios coincide there and the upper
constructions apply verbatim; on the family $D(k)$ the survivor ratio
is exactly $2(k+1)/(3k)\to2/3$. The four $\NC$ cells inherit vanishing
from $\tau^-\le\tau^+$, with unique-champion witnesses under both
conventions (e.g.\ $(4k-1)/(k(k+1))\to0$ on $H(k,2)$, survivor scale).
\end{proof}

% =====================================================================
% Larger budgets (lower bounds, simultaneous guarantees, apex lift) --
% separate file.
% tauminus-section.tex -- larger budgets: nestedness, lower bounds at
% every exact budget, simultaneous guarantees, apex lift.

\section{Larger budgets}\label{sec:budgets}

The constants of Theorem~\ref{thm:main} concern single-vertex attacks.
This section records what carries over to attack sets of a fixed size
$b\ge2$: the two lower bounds do, the four vanishing directions and the
absence of a simultaneous guarantee do as well, and over all connected
graphs nothing improves. Throughout, $N=n-b\ge1$ is the number of
surviving vertices, every attack set $S$ with $|S|=b$ leaves a profile
$\sigma(G-S)$ whose parts sum to $N$, and we keep the notation
$Q(S)=\sum_i\binom{c_i}2$ and $s(S)=c_1(S)$. The survivor-based damages
are $\dam'_\PC(S)=\binom N2-Q(S)$ and $\dam'_\LCC(S)=N-s(S)$; the
inclusive damages of Definition~\ref{def:objectives} exceed them by the
constants $\binom n2-\binom N2$ and $b$ respectively, so for a fixed
budget the champion sets do not depend on the convention.

\subsection{Optimal sets at consecutive budgets}\label{sec:nested}

A vertex optimal at budget one need not belong to any optimal attack set
of size two. Define the \emph{budget-extension ratio} of $G$ (for the
$\PC$ objective, optimistic at both levels) as
\[
\beta(G)=\max_{v\in\Top_\PC(G,1)}\ \max_{w\neq v}
\dam_\PC(\{v,w\})\ \big/\ \max_{|S|=2}\dam_\PC(S).
\]
Already for the path $P_5$ the $b{=}1$ champion, the centre, belongs to
\emph{no} optimal pair, and over all connected graphs with $4\le n\le7$
the minimum of $\beta$ is exactly $3/4$, attained at $n=7$ by the
complete bipartite graph $K_{2,4}$ with a pendant vertex attached to a
vertex of the larger side (best champion-containing pair $15$ versus
optimal pair $20$; Table~\ref{tab:comp}). Optimal sets at
consecutive budgets are thus not nested, and the results below are
statements about each budget separately.

\subsection{Lower bounds at every exact budget}
\label{sec:budget-ext}

Nothing in the proofs of Propositions~\ref{prop:pc-lcc}
and~\ref{prop:lcc-pc} uses that a single vertex is deleted.

\begin{corollary}\label{cor:budget}
Let $G$ be connected on $n\ge3$ vertices and let $1\le b\le n-1$. Under
either damage convention, on every cell whose target optimum is
positive, every $\PC$-optimal attack set $S$ of size $b$ satisfies
$\dam_\LCC(S)\ge(2-\sqrt2)\,M_\LCC(G,b)$, and every $\LCC$-optimal
attack set $S$ of size $b$ satisfies $\dam_\PC(S)\ge\tfrac23\,M_\PC(G,b)$.
Consequently the pessimistic and optimistic transfer ratios in the
directions $\PC\to\LCC$ and $\LCC\to\PC$ are bounded below by
$2-\sqrt2$ and $2/3$ at every exact budget.
\end{corollary}

\begin{proof}
Lemma~\ref{lem:envelope} is a statement about partitions of $N$ with a
prescribed largest part and applies verbatim to every profile
$\sigma(G-S)$: $\binom{s(S)}2\le Q(S)\le$ the right-hand side
of~\eqref{eq:env} evaluated at $s=s(S)$. Let $\ell$ be $\LCC$-optimal,
so that $a:=s(\ell)=\min_{|S|=b}s(S)$, and let $p$ be $\PC$-optimal, so
that $Q(p)=\min_{|S|=b}Q(S)$; put $B:=s(p)\ge a$.

\emph{Survivor scale, $\PC\to\LCC$.} Exactly as in the proof of
Theorem~\ref{thm:tauminus},
\[
B(B-1)=2\tbinom B2\le 2Q(p)\le 2Q(\ell)\le 2F(N,a)\le R'(R'-1),
\qquad R'=N-(2-\sqrt2)(N-a):
\]
the last inequality is $E_1\ge0$ in the regime $a\le N/2$ (where
$2F(N,a)\le N(a-1)$) and the identity for $E_2\ge0$ in the regime
$a\ge N/2$ (where $2F(N,a)\le a(a-1)+d(d-1)$, $d=N-a$); both involve
only the real parameters $N$ and $a$ and hold for every $N\ge2$ and
$1\le a\le N-1$. Since $R'\ge1$ and $x\mapsto x(x-1)$ is increasing on
$[1,\infty)$, $B\le R'$, hence
$\dam'_\LCC(p)/M'_\LCC=(N-B)/(N-a)\ge2-\sqrt2$ whenever $N-a>0$, i.e.\
whenever the target optimum is positive.

\emph{Survivor scale, $\LCC\to\PC$.} Every attack set $S$ has
$s(S)\ge a$ and therefore $Q(S)\ge\binom a2$, so
$M'_\PC\le\binom N2-\binom a2$, while $Q(\ell)\le F(N,a)$. The two regime
bounds $N/(N+a-1)\ge2/3$ for $a\le N/2$ and $d(N-2d+1)\ge0$ for
$d=N-a\le N/2$ follow as in the proof of Theorem~\ref{thm:tauminus},
again involving only $N$ and $a$.

\emph{Inclusive scale.} For a fixed budget the inclusive damages are the
survivor damages plus a constant $k\ge0$ that is the same for every
attack set ($k=b$ for $\LCC$, $k=\binom n2-\binom N2$ for $\PC$). If
$0\le x\le y$ and $y>0$ then $(x+k)/(y+k)\ge x/y$, which transfers both
bounds; if the survivor target optimum is $0$ while the inclusive one is
positive, every attack set attains the same inclusive target damage $k$
and the ratio equals $1$.

Finally, no particular champion was selected anywhere: every
$\PC$-optimal set has the minimum value of $Q$ and every $\LCC$-optimal
set leaves largest part $a$, so the bounds hold for each tied champion
separately.
\end{proof}

The corollary makes no sharpness claim for $b\ge2$. Over all connected
graphs the two constants are in fact optimal at every fixed budget, and
the four $\NC$ cells remain vanishing there, because the budget-one
witness families lift to any fixed $b$ by adding universal vertices
(Section~\ref{sec:apex}). What remains open is the behaviour
\emph{within} graph classes that the lift leaves --- trees and
bounded-degree graphs at $b\ge2$ in particular --- and the regime of
budgets growing with $n$.

\subsection{Simultaneous guarantees}
\label{sec:qsplit}

The transfer ratios evaluate the reuse of one objective's champion. A
weaker requirement is that \emph{some} attack set --- not necessarily a
champion of anything --- retain a positive fraction of every optimum in
a family at once. Vanishing transfer to and from $\NC$ does not by
itself exclude such a compromise. The following proposition shows that
for the pairs containing $\NC$ no uniform compromise exists, even on
trees, while the pair $\PC,\LCC$ admits one.

\begin{definition}\label{def:q}
Let $\mathcal C$ be a non-empty finite family of objectives, and let
$(G,b)$ be a cell with $M_u(G,b)>0$ for every $u\in\mathcal C$. Put
\[
q_{\mathcal C}(G,b)=\max_{|S|=b}\ \min_{u\in\mathcal C}
\frac{\dam_u(S)}{M_u(G,b)}\in[0,1].
\]
Thus $q_{\mathcal C}(G,b)=1$ iff the family has a common champion on
$(G,b)$.
\end{definition}

\begin{proposition}\label{prop:qsplit}
Under either damage convention, restricted to cells where every
objective of the family has positive optimum:
\begin{enumerate}
\item[(a)] $q_{\{\PC,\LCC\}}(G,1)\ge 2/3$ for every connected graph
$G$;
\item[(b)] $\inf_G q_{\{\PC,\NC\}}(G,1)=\inf_G q_{\{\LCC,\NC\}}(G,1)
=\inf_G q_{\{\PC,\LCC,\NC\}}(G,1)=0$, the infima being approached
along a family of trees: for $k\ge2$ and $m=k^2$, let $H_k$ be the path
$v_0v_1\cdots v_{2m}$ with $k$ pendant leaves attached at $v_0$; then
\[
q_{\{\PC,\NC\}}(H_k,1)\le\frac2k+\frac1{2k^2},\qquad
q_{\{\LCC,\NC\}}(H_k,1)\le\frac{k+1}{k^2+1},
\]
and $q_{\{\PC,\LCC,\NC\}}(H_k,1)\le q_{\{\LCC,\NC\}}(H_k,1)$.
\end{enumerate}
\end{proposition}

\begin{proof}
(a) Take $S=\{\ell\}$ for any $\LCC$-champion $\ell$. Its $\LCC$ ratio
is $1$ and, by Theorem~\ref{thm:tauminus} (pessimistic $\LCC\to\PC$
constant, valid for every tied champion and under both conventions),
its $\PC$ ratio is at least $2/3$.

(b) $H_k$ has $n=2m+1+k$ vertices. Deleting $v_0$ leaves $k$ isolated
leaves and the path $v_1\cdots v_{2m}$, so $\dam_\NC(v_0)=k$ and
$M_\NC(H_k,1)\ge k$; deleting any other vertex leaves at most two
components, so every $v\ne v_0$ has $\NC$ ratio at most $1/k$ (a leaf
has ratio $0$). It therefore suffices to bound the $\PC$ and $\LCC$
ratios of $v_0$, and for that a lower bound on the target optimum
suffices: we use the damage of the middle vertex $v_m$, whose deletion
leaves the two components $\{v_0,\dots,v_{m-1}\}\cup\{\text{leaves}\}$
of order $m+k$ and $\{v_{m+1},\dots,v_{2m}\}$ of order $m$.

\emph{Inclusive scale.} Expanding the binomials,
$\dam_\PC(v_0)=\binom n2-\binom{2m}2=2mk+2m+\tfrac{k^2+k}2$ and
$\dam_\PC(v_m)=\binom n2-\binom{m+k}2-\binom m2=m^2+mk+2m+k$. With
$m=k^2$,
\[
\frac{\dam_\PC(v_0)}{\dam_\PC(v_m)}
=\frac{2k^3+\tfrac52k^2+\tfrac12k}{k^4+k^3+2k^2+k}
\le\frac2k+\frac1{2k^2},
\]
since $(\tfrac2k+\tfrac1{2k^2})(k^4+k^3+2k^2+k)
=2k^3+\tfrac52k^2+\tfrac92k+3+\tfrac1{2k}$ exceeds the numerator.
For $\LCC$, $\dam_\LCC(v_0)=n-2m=k+1$ and $\dam_\LCC(v_m)=n-(m+k)=m+1$,
so the $\LCC$ ratio of $v_0$ is at most $(k+1)/(k^2+1)$.

\emph{Survivor scale.} With $N=n-1=2m+k$: $\dam'_\PC(v_0)=2mk+\tfrac{k^2-k}2$
and $\dam'_\PC(v_m)=m^2+mk$, whose ratio at $m=k^2$ is
$(2k^2+\tfrac12k-\tfrac12)/(k^2(k+1))\le2/k$; and
$\dam'_\LCC(v_0)=k$, $\dam'_\LCC(v_m)=m$, ratio $1/k\le(k+1)/(k^2+1)$.
Both are below the inclusive bounds.

Hence $q_{\{\PC,\NC\}}(H_k,1)\le\max\{1/k,\ 2/k+1/(2k^2)\}$ and
$q_{\{\LCC,\NC\}}(H_k,1)\le\max\{1/k,\ (k+1)/(k^2+1)\}$, which are the
stated bounds; the triple is bounded by any of its pairs. All three
optima are positive on $H_k$, and the bounds tend to $0$.
\end{proof}

\begin{remark}
The proposition concerns exact budget one and the positive-optimum
domain; $2/3$ is a lower bound for $q_{\{\PC,\LCC\}}$ and is not claimed
to be optimal (a compromise vertex may beat both champions). The
question is one of simultaneous approximation, for which both existence
theorems with tight constants and their absence in other models are
known
\cite{stein1997bicriteria,azar2002allnorm,goel2006simultaneous,goel2005majorization};
here the objectives are the canonical network damages on
graph-constrained profiles, the fragment count is not a norm, and the
negative side already occurs on trees. The feasible-set
characterisations of Schoot Uiterkamp~\cite{schoot2025simultaneous}
require least elements uniformly over shifts under convexity or
integral hole-freeness assumptions; an unshifted common optimum does not
imply those hypotheses, and no characterisation of deletion-profile
regions by submodular base polyhedra is asserted here.
\end{remark}

\begin{remark}\label{rem:portfolio}
Fix $\mathcal C=\{\PC,\LCC,\NC\}$, a connected $G$, and an exact
budget $1\le b<n$ on which all three optima are positive. For a
non-empty set $\mathcal P$ of feasible $b$-vertex attacks, put
\[
 \Gamma_{\mathcal C}(\mathcal P)=
 \min_{u\in\mathcal C}\max_{S\in\mathcal P}
 \frac{\dam_u(S)}{M_u(G,b)},\qquad
 c_\alpha(G,b)=\min\{|\mathcal P|:\Gamma_{\mathcal C}(\mathcal P)\ge\alpha\},
 \quad 0<\alpha\le1.
\]
Different objectives may be covered by different scenarios. Choose any
$\LCC$-champion $S_L$ and any $\NC$-champion $S_N$. The
budget-independent $\LCC\to\PC$ lower bound gives the coverage vector
of $\{S_L,S_N\}$ componentwise at least $(2/3,1,1)$ in the order
$(\PC,\LCC,\NC)$. Hence $c_{2/3}(G,b)\le2$ under either convention,
uniformly over ties. Also $c_\alpha(G,b)=1$ iff
$q_{\mathcal C}(G,b)\ge\alpha$; in particular $c_1=1$ iff a common
champion exists. Proposition~\ref{prop:qsplit} gives
$q_{\mathcal C}(H_k,1)\le(k+1)/(k^2+1)\le3/5<2/3$ for $k\ge2$,
so $c_{2/3}(H_k,1)=2$, and for every fixed $0<\alpha\le1$,
sufficiently large $k$ gives $c_\alpha(H_k,1)\ge2$. By the apex lift
(Section~\ref{sec:apex}) the same holds for $H_k\vee K_{b-1}$ at every
fixed $b\ge2$. Necessity of two tests on every graph, on trees for
$b\ge2$, or an optimal coverage constant for two tests, is not asserted.

Each scenario starts from the original graph: two tests of $b$ failures
are not one attack on their union. The guarantee concerns topological
damage, not throughput, latency or service continuity, and it requires
exact champions; an approximate optimiser does not inherit it. At any
fixed $b$ the deterministic triple can be optimised by enumerating
$\binom nb$ attacks; efficient computation when $b$ is part of the input
is a separate question.
\end{remark}

\subsection{Fixed budgets over all connected graphs}
\label{sec:apex}

The lower bounds of Corollary~\ref{cor:budget} hold at every exact
budget. The following construction shows that, over \emph{all}
connected graphs, they cannot be improved at any fixed budget, and that
the vanishing and the simultaneous-guarantee phenomena persist. The
gadget --- universal vertices that every damaging attack must delete ---
is standard in hardness reductions for critical-node problems; what
matters here is that it preserves champion sets under both conventions,
survivor-normalised transfer ratios and simultaneous guarantees
exactly, and inclusive ratios up to an additive offset that leaves the
limiting constants unchanged along the witness families.

\begin{lemma}\label{lem:apex}
Let $G$ be connected on $n\ge2$ vertices, let $b\ge2$, and let $J=G\vee
K_{b-1}$ be obtained by adding a set $A$ of $b-1$ new vertices, each
adjacent to every vertex of $G$ and to each other. Consider exact budget
$b$ on $J$, so that $N=n-1$ vertices survive. Then:
\begin{enumerate}
\item[(i)] if $S\not\supseteq A$, a surviving apex is adjacent to every
other survivor, $J-S$ is connected, and
$\dam'_\PC(S)=\dam'_\LCC(S)=\dam_\NC(S)=0$;
\item[(ii)] if $S\supseteq A$ then $S=A\cup\{v\}$ with $v\in V(G)$ and
$J-S=G-v$, so $\sigma(J-S)=\sigma(G-v)$ with the same number $N$ of
survivors;
\item[(iii)] for every objective $u\in\{\PC,\LCC,\NC\}$ with
$M'_u(G,1)>0$ (survivor scale): $M'_u(J,b)=M'_u(G,1)$ and, under either
convention,
\[
\Top_u(J,b)=\bigl\{A\cup\{v\}: v\in\Top_u(G,1)\bigr\},
\]
all ties included.
\end{enumerate}
Consequently, on the survivor scale, for all $u,w\in\{\PC,\LCC,\NC\}$
with $M'_u(G,1)>0$ and $M'_w(G,1)>0$,
$\tau'^{\pm}_{u\to w}(J,b)=\tau'^{\pm}_{u\to w}(G,1)$, and for every
non-empty $\mathcal C\subseteq\{\PC,\LCC,\NC\}$ whose members all have
positive survivor optimum on $(G,1)$,
$q'_{\mathcal C}(J,b)=q'_{\mathcal C}(G,1)$. If $G$ has a cut vertex,
all three survivor optima of $(G,1)$ are positive and the identities
hold for every pair and every family.
\end{lemma}

\begin{proof}
(i) and (ii) are immediate from the construction. For (iii), by (i)
every attack set missing an apex has survivor damage $0$ under each
objective, and by (ii) the sets containing $A$ realise exactly the
profiles $\sigma(G-v)$, $v\in V(G)$, on the same survivor count $N$.
Hence $M'_u(J,b)=\max\{0,M'_u(G,1)\}=M'_u(G,1)>0$, and the maximisers
are precisely the sets $A\cup\{v\}$ with $v\in\Top_u(G,1)$: no
zero-damage set attains a positive maximum. At a fixed budget the two
conventions differ by a constant, so the champion sets are the same
under both. On the survivor scale the ratio of $S=A\cup\{v\}$ under
any target equals the ratio of $v$ on $(G,1)$, numerator and denominator
being unchanged; maximising or minimising over the lifted champion set
is therefore maximising or minimising over $\Top_u(G,1)$. For
$q'_{\mathcal C}$, sets missing an apex have inner minimum $0$, sets
$A\cup\{v\}$ have the inner minimum of $v$, and $q'_{\mathcal C}(G,1)
\ge0$; the maxima coincide.
\end{proof}

\begin{proposition}\label{prop:apex}
Fix an exact budget $b\ge1$. Under either damage convention, on cells
where the target optimum is positive:
\begin{enumerate}
\item[(a)] $\inf_G\tau^{\pm}_{\PC\to\LCC}(G,b)=2-\sqrt2$ and
$\inf_G\tau^{\pm}_{\LCC\to\PC}(G,b)=2/3$, the infima being over
connected graphs $G$ and approached along explicit families;
\item[(b)] the four directions involving $\NC$ are vanishing at budget
$b$;
\item[(c)] over cells on which every member of the family has positive
optimum, the infima over $G$ of $q_{\{\PC,\NC\}}(G,b)$,
$q_{\{\LCC,\NC\}}(G,b)$ and $q_{\{\PC,\LCC,\NC\}}(G,b)$ are all
$0$, whereas $q_{\{\PC,\LCC\}}(G,b)\ge2/3$ for every such $G$.
\end{enumerate}
\end{proposition}

\begin{proof}
For $b=1$ these are Theorems~\ref{thm:main} and~\ref{thm:tauminus} and
Proposition~\ref{prop:qsplit}. Let $b\ge2$. The lower bounds in (a) and
the bound $q_{\{\PC,\LCC\}}\ge2/3$ (any $\LCC$-champion) are
Corollary~\ref{cor:budget}; the lower bound $0$ in (b) and (c) is
trivial. For the upper bounds, every witness family of
Section~\ref{sec:proofs} and the family $H_k$ of
Proposition~\ref{prop:qsplit} consists of graphs with a cut vertex, so
Lemma~\ref{lem:apex} applies to each member $G_k$ and its lift
$J_k=G_k\vee K_{b-1}$.

\emph{Survivor scale.} By the lemma, $\tau'^{\pm}(J_k,b)=\tau'^{\pm}(G_k,1)$
and $q'(J_k,b)=q'(G_k,1)$ for every $k$; the limits along the lifted
families are the budget-one values $2-\sqrt2$, $2/3$ and $0$, which gives
(a)--(c).

\emph{Inclusive scale.} For $S=A\cup\{v\}$ the inclusive damages of
$(J_k,b)$ and of $(G_k,1)$ differ by the constants
\[
\delta_\LCC=b-1,\qquad
\delta_\PC=\tbinom{n+b-1}2-\tbinom n2=(b-1)n+\tbinom{b-1}2,\qquad
\delta_\NC=0
\]
($n=|V(G_k)|$), and so do the optima; hence for the lifted champions
\[
\tau^{\pm}_{u\to w}(J_k,b)=
\frac{M_w(G_k,1)\,\tau^{\pm}_{u\to w}(G_k,1)+\delta_w}{M_w(G_k,1)+\delta_w},
\]
which has the same limit as $\tau^{\pm}(G_k,1)$ whenever
$\delta_w/M_w(G_k,1)\to0$. For a fixed $b$ this holds in every cell
used: targets $\NC$ have $\delta_\NC=0$ (the directions $\PC\to\NC$ on
$H(k,\lfloor k/2\rfloor)$ and $\LCC\to\NC$ on $D(k)$); target $\LCC$ has
$\delta_\LCC=b-1$ constant while $M_\LCC(G_k,1)\to\infty$ (equal to
$a+1$ on the $2-\sqrt2$ family of order $2a+1$, and of order $k^2$ on the
hub-and-tail family, whose path has $k^2$ vertices); target $\PC$ has
$\delta_\PC=O(n)$ while $M_\PC(G_k,1)=\Theta(n^2)$ (equal to
$(3k^2+7k+4)/2$ on $D(k)$, order $2k+3$, and to $k^2+3k+1$ on $H(k,2)$,
order $2k+2$). For (c), take $\mathcal C\ni\NC$: a set missing an apex has
$\NC$-damage $0$ and inner minimum $0$; for $S=A\cup\{v\}$ each ratio
satisfies $(x_u+\delta_u)/(M_u+\delta_u)\le x_u/M_u+\delta_u/M_u$, so
$q_{\mathcal C}(J_k,b)\le q_{\mathcal C}(H_k,1)+\max_u\delta_u/M_u(H_k,1)$.
On $H_k$ (order $2k^2+k+1$), $M_\PC(H_k,1)\ge\dam_\PC(v_m)=k^4+k^3+2k^2+k$
and $M_\LCC(H_k,1)\ge k^2+1$, so the correction is $O(1/k^2)$ and the
bound tends to $0$.
\end{proof}

\begin{remark}\label{rem:apex}
The statement is for each \emph{fixed} $b$; nothing is claimed for
budgets growing with $n$. The infima are over all connected graphs: a
class that contains a lifted witness family inherits the corresponding
cell, while a class that excludes it is not decided by this argument.
The lift leaves the class of trees as soon as $b\ge2$; if $G$ is a tree
then $J$ has treewidth exactly $b$ ($A$ together with any edge of $G$ is
a $(b+1)$-clique, and adding $A$ to every bag of a width-one
decomposition gives width $b$). Thus for every fixed $b$ the
$\NC$-containing $q$-infima are already $0$ within treewidth at most
$b$; nothing follows from this for the other cells within a class, and
whether trees or bounded-degree graphs admit better constants at
$b\ge2$ remains open. The value $2/3$ for $q_{\{\PC,\LCC\}}$ is a lower
bound, not claimed optimal.
\end{remark}

% =====================================================================
\section{Computational methods}\label{sec:method}

The finite-size separation thresholds in Theorem~\ref{thm:thresholds}
were established by exhaustive enumeration. Additional computations
checked the transfer bounds and witness families on the graph classes
and random samples listed in Table~\ref{tab:comp}. Connected graphs and
trees of each order were generated with \texttt{geng} from the nauty
package, and the connected-graph counts were checked against OEIS
A001349. The separation counts were reproduced by two independently
written implementations, with identical results.

Ratios were compared in exact rational arithmetic. The statements about
the failure probability were derived from exact polynomials with
isolated real roots, so the reported comparisons do not depend on
floating-point arithmetic. Computer search guided the construction of
the witness families in Section~\ref{sec:proofs}, whose properties are
proved analytically.

Table~\ref{tab:comp} summarises the outcome. The code, raw outputs,
witness graphs in graph6 format, and a map from each statement to the
corresponding computation are provided as supplementary material.

\begin{table}[htb]
\centering\footnotesize
\setlength{\tabcolsep}{4pt}
\begin{tabular}{@{}>{\raggedright\arraybackslash}p{.30\textwidth}
>{\raggedright\arraybackslash}p{.27\textwidth}
>{\raggedright\arraybackslash}p{.37\textwidth}@{}}
\toprule
Statement & Graphs examined & Outcome\\
\midrule
Thm.~\ref{thm:thresholds}(i)--(iii): pairwise separation & all connected graphs, $4\le n\le9$ & first disjoint pairs at $n=7$ ($\LCC/\NC$: 8 of 853), $n=8$ ($\PC/\NC$: 48 of 11{,}117), $n=9$ ($\PC/\LCC$: 129 of 261{,}080); witnesses as in Theorem~\ref{thm:thresholds}\\
Thm.~\ref{thm:thresholds}(iv): triple separation & all connected graphs, $n\le10$; trees, $n\in\{11,12,14\}$ & none for $n\le10$ (11{,}716{,}571 graphs at $n=10$); trees: 3 of 235 ($n=11$), 3 of 551 ($n=12$), 34 of 3159 ($n=14$)\\
Thm.~\ref{thm:regime}: regime reversal & 29 connected graphs, $n\le5$; the graph $R_6$ & no leadership change for $n\le5$; unique root $p^*=3-\sqrt5$\\
Thms.~\ref{thm:main}, \ref{thm:tauminus}: ratios at $b=1$, both conventions & all connected graphs, $n\le9$; 24{,}934 random connected graphs, $10\le n\le60$ (132 with non-trivial ties) & no pessimistic ratio below $2-\sqrt2$ or $2/3$; minima at $n=9$: $\PC\to\LCC$ $4/5$ (inclusive), $3/4$ (survivor), at $D(3)$; $\LCC\to\PC$ $23/26$, $5/6$, at a unicyclic graph; optimistic minima $12/13$ ($\LCC\to\PC$) and $1/3$ ($\LCC\to\NC$), both at $D(3)$\\
Cor.~\ref{cor:budget}: ratios at $b\le3$, both conventions, two engines & all connected graphs, $n\le8$; 916 random connected graphs, $9\le n\le30$ & no pessimistic ratio below the constants; smallest at $b=2$: $2/3$ ($\PC\to\LCC$) and $3/4$ ($\LCC\to\PC$), survivor scale, $n=8$\\
Prop.~\ref{prop:sharp}(2): champion uniqueness on the $2-\sqrt2$ family & $a\in\{8,15,30,60,120\}$ & confirmed; ratios $2/3,\ 11/16,\ 19/31,\ 37/61,\ 72/121$\\
Prop.~\ref{prop:qsplit}: $q$-bounds on $H_k$ & $k\le30$, both conventions & bounds hold with exact fractions\\
Remark~\ref{rem:portfolio}: coverage by two tests & all connected graphs $n\le7$, $b\le3$; $H_k$, $k\le12$ & coverage $\ge(2/3,1,1)$ in every case; $c_{2/3}(H_k,1)=2$\\
Lemma~\ref{lem:apex}: apex lift & all connected graphs $n\le7$ at $b\in\{2,3\}$, $n\le6$ at $b=4$; all witness families & champion sets, survivor ratios and $q$ preserved as stated\\
Nestedness ratio $\beta$ (Section~\ref{sec:nested}) & all connected graphs, $4\le n\le7$ & minimum $3/4$ ($n=7$); first non-nested case $P_5$\\
\bottomrule
\end{tabular}
\caption{Summary of the computations. ``Both conventions'' refers to the inclusive damages of Definition~\ref{def:objectives} and the survivor-based damages of Remark~\ref{rem:normalisation}.}
\label{tab:comp}
\end{table}

% =====================================================================
\section{Open problems}\label{sec:outlook}

The main open question concerns budgets $b\ge2$ within restricted
classes. Corollary~\ref{cor:budget} gives the constants $2-\sqrt2$ and
$2/3$ as lower bounds at every budget, and Section~\ref{sec:apex} shows
that over all connected graphs they cannot be improved; but the apex
construction leaves the class of trees and raises the treewidth, so it
says nothing about whether the constants improve on trees or on
bounded-degree graphs, or whether the $\NC$ directions remain vanishing
there. The same question arises for budgets proportional to $n$, which
no fixed-$b$ statement covers.

Two further questions seem natural. The three objectives are the
projections $\bigl(\sum_i\binom{c_i}2,\,c_1,\,m\bigr)$ of the component
profile; which other profile functionals admit uniform transfer
constants, and under which normalisations beyond the two settled in
Theorem~\ref{thm:tauminus}? And what does the full set of regime
breakpoints of a graph look like --- how many are there, of what
algebraic degree, and how do they depend on the topology? We have no
answers beyond the results reported here.

% =====================================================================
\FloatBarrier
\section*{Acknowledgements}
Generative AI tools, including GPT-based assistants, were used during the preparation of this manuscript to explore mathematical ideas and refine the exposition. The author takes full responsibility for all mathematical claims, proofs, computational results, and conclusions.

\bibliographystyle{plainurl}
\bibliography{refs}

\end{document}